%% file: Bernardin.tex
\documentclass[graybox]{svmult}

\usepackage{mathptmx}       
\usepackage{helvet}         
\usepackage{courier}        
\usepackage{makeidx}         
\usepackage{graphicx}        
\usepackage{multicol}        

\makeindex             

\usepackage{amsmath}
\usepackage{amssymb}
\usepackage{amsfonts}

\newcommand\LL{{\mathbb L}}

\newcommand\RR{{\mathbb R}}
\newcommand\TT{{\mathbb T}}
\newcommand\ZZ{{\mathbb Z}}

\newcommand{\ri}{\mathrm{i}}

\newcommand{\mc}[1]{{\mathcal #1}}

\newcommand{\bb}[1]{{\mathbb #1}}

\begin{document}

\title*{Superdiffusion of energy in Hamiltonian systems perturbed by a conservative noise}
\author{C\'edric Bernardin}
\institute{C\'edric Bernardin \at Universit\'e de Nice Sophia-Antipolis, Laboratoire J.A. Dieudonn\'e, UMR CNRS 7351, Parc Valrose, 06108 Nice cedex 02- France,
 \email{cbernard@unice.fr}}
%
%
\maketitle

\abstract{We review some recent results on the anomalous diffusion of energy in systems of 1D coupled oscillators and we revisit the role of momentum conservation. }

\section{Introduction}
Transport properties of one-dimensional Hamiltonian systems consisting
of coupled oscillators on a lattice have been the subject of many theoretical and numerical studies, see the review papers~\cite{BLR,D,LLP}. Despite many efforts, our knowledge of the fundamental mechanisms necessary and/or sufficient to have a normal diffusion remains very limited.

Consider a $1$-dimensional chain of oscillators indexed by $x \in \ZZ$, whose formal Hamiltonian is given by
\begin{equation*}
{\mc H} = \sum_{x \in \ZZ} \left[ \frac{p_x^2}{2} +V(r_x) \right],
\end{equation*} 
where $r_x=q_{x+1}-q_x$ is the ``deformation'' of the lattice, $q_x$  being the displacement of the atom $x$ from its equilibrium position and $p_x$ its momentum. The interaction potential $V$ is a smooth positive function growing at infinity fast enough. The energy $e_x$ of atom $x \in \ZZ$ is defined by 
$$e_x = \frac{p_x^2}{2} + V(r_x).$$ 

Our goal is to understand the macroscopic energy diffusion properties for the corresponding Hamiltonian dynamics 
\begin{equation*}
\cfrac{dr_x}{dt} = p_{x+1} -p_x, \quad \cfrac{dp_x}{dt} = V'(r_{x}) -V'(r_{x-1}), \quad x\in \ZZ.
\end{equation*}
Under suitable conditions on $V$, the infinite dynamics is well defined for a large class of initial conditions. 

Apart from the total energy $\sum_x e_x$, observe that the total momentum $ \sum_x p_x$ and the total deformation $\sum_x r_x$ of the lattice are formally conserved. This is a consequence of the following microscopic continuity equations:
\begin{eqnarray}
&&\frac{de_x}{dt} + \nabla [j_{x-1,x}^e]=0, \quad j_{x,x+1}^e= -p_{x+1} V' (r_x),\\
&&\frac{dp_x}{dt} + \nabla[ -V' (r_{x-1})]=0,\\
&&\frac{dr_x}{dt} + \nabla[ -p_x ]=0.
\end{eqnarray}
The function $j_{x,x+1}^e$ is the current of energy going from $x$ to $x+1$. The main open problem (\cite{LP}, \cite{Sz}) concerning the foundation of statistical mechanics based on classical mechanics is precisely to show that the three quantities above are the only quantities which are conserved by the dynamics. In some sense, it means that the dynamics, evolving on the manifold defined by fixing the total energy, the total momentum and the total deformation, is ergodic. Of course, the last sentence does not make sense since we are in infinite volume and $\sum_x e_x, \sum_x p_x$ and $\sum_x r_x$ are typically infinite. Nevertheless, an alternative meaningful definition will be proposed and discussed in Section \ref{sec:ergodicity}.    

Numerical simulations provide a strong evidence of the fact that one dimensional chains of anharmonic oscillators conserving momentum are {\footnote{See however the coupled-rotor model which displays normal behavior (see \cite{LLP}, Section 6.4). This is probably due to the fact that the position space is compact. }} superdiffusive. It shall be noticed that there is no explanation of this, apart from heuristic considerations, and that some models which do not conserve momentum can also display anomalous diffusion of energy (see \cite{GDL}). 

%

An interesting area of current research consists in studying this problem for hybrid models where a stochastic perturbation is superposed to the deterministic evolution. Even if the problem is considerably simplified, several open challenging questions can be addressed for these systems. The first benefit of the introduction of stochasticity in the models is to guarantee the ergodicity that we are not able to show for purely deterministic systems. The added noise must be carefully chosen in order not to destroy the conservation laws we are interested in. In particular, the noise shall conserve energy. But we will consider a noise conserving also some of the other quantities conserved by the underlying Hamiltonian dynamics, e.g. the momentum, the deformation or any linear combination of them. 

The paper is organized as follows. In Section \ref{sec:ergodicity} we discuss the problem of the ergodicity of the infinite dynamics mentioned above and the possible stochastic perturbations we can add to the deterministic dynamics to obtain ergodic dynamics. In Section \ref{sec:harmonic} we review some results obtained in the context of harmonic chains perturbed by a conservative noise and we discuss the case of anharmonic chains in the last section.

\section{Ergodicity}
\label{sec:ergodicity}

Let us first generalize the models introduced above (\cite{BS}). Let $U$ and $V$ be smooth positive potentials growing at infinity fast enough and let ${\mc H}:={\mc H}_{U,V}$ be the Hamiltonian 
$${\mc H}_{U,V} =\sum_{x \in \ZZ} \left[ U(p_x) +V(r_x) \right].$$
The corresponding Hamiltonian dynamics satisfy 
\begin{equation}
\label{eq:Ham-gen}
\frac{dr_x}{dt}= U'(p_{x+1}) -U' (p_{x}), \quad \frac{dp_x}{dt} =V'(r_x) -V'(r_{x-1}), \quad x \in \ZZ.
\end{equation}
The energy of particle $x$ is defined by $e_x = U(p_x) + V(r_x)$. The three formal quantities $\sum_x e_x$, $\sum_x r_x$ and $\sum_x p_x$ are conserved by the dynamics. 
The fundamental question we address in this section is: are they the only ones?
In finite volume, i.e. replacing the lattice $\ZZ$ by a finite box $\Lambda$, this would correspond to the usual notion of ergodicity for Hamiltonian flows with a finite number of degrees of freedom. But since we consider the dynamics in infinite volume the notion of conserved quantity has to be properly defined. The way we follow to attack the problem is to detect the existence of a non-trivial conserved quantity through the existence of a non-trivial invariant state for the infinite dynamics.   

Let $\Omega=(\RR \times \RR)^{\ZZ}$ be the phase space of the dynamics and let us denote a typical configuration by $\omega=(r,p) \in \Omega$. For simplicity we assume that for any $(\beta, \lambda, \lambda' ) \in (0, +\infty) \times \RR \times \RR$, the partition function 
$$Z(\beta, \lambda, \lambda')= \int_{\RR \times \RR} e^{-\beta[U(a) +V(b)] -\lambda b - \lambda' a } da \,  db$$ 
is finite. Let $\mu_{\beta, \lambda,\lambda'}$ be the product Gibbs measures on $\Omega$ defined by 
\begin{equation*}
d\mu_{\beta, \lambda,\lambda'} (\omega) = \prod_{x \in \ZZ} \frac{1}{Z(\beta, \lambda,\lambda')} \exp \left[ - \beta [U(p_x) + V(r_x)] - \lambda r_x - \lambda' p_x \right] dr_x dp_x.
\end{equation*}
We assume that (\ref{eq:Ham-gen}) is well defined for a subset $\Omega_{\beta, \lambda,\lambda'}$ of full measure with respect to $\mu_{\beta, \lambda,\lambda'}$, that the latter is invariant for (\ref{eq:Ham-gen}), and that it is possible to define a strongly continuous semigroup in ${\bb L}^2 (\mu_{\beta,\lambda,\lambda'})$ with formal generator 
 \begin{equation*}
 {\mc A}_{U,V} = \sum_{x \in \ZZ} \left[(U'(p_{x+1}) -U'(p_x)) \partial_{r_x} + (V'(r_x) -V'(r_{x-1})) \partial_{p_x} \right].
 \end{equation*}
All that can be proved under suitable assumptions on $U$ and $V$ (\cite{FFL}, \cite{BObook}). 

In order to explain what is meant by ergodicity of the infinite volume dynamics 
we need to introduce some notation. 
For any topological space $X$ equipped with its Borel $\sigma$-algebra we 
denote by ${\mc P} (X)$ the convex set of probability measures on $X$.  
The relative entropy $H(\nu|\mu)$ of $\nu \in {\mc P} (X)$ with respect to 
$\mu \in {\mc P} (X)$ is defined as
\begin{equation}
H(\nu | \mu) = \sup_{\phi} \left\{ \int \phi \, d\nu - 
\log \left( \int e^{\phi} \, d\mu \right) \right\},
\end{equation}
where the supremum is carried over all bounded measurable functions $\phi$ on $X$. 

Let $\theta_x, x \in \ZZ$, be the shift by $x$: $(\theta_x \omega)_z=\omega_{x+z}$. 
For any function $g$ on $\Omega$, $\theta_x g$ is the function such that $(\theta_xg)(\omega) 
= g(\theta_x \omega)$. For any probability measure $\mu \in {\mc P} (\Omega)$, 
$\theta_x \mu \in {\mc P} (\Omega)$ is the probability measure such that, 
for any bounded function $g: \Omega \to \RR$, 
it holds $\int_\Omega g \, d (\theta_x \mu)= \int_\Omega \theta_x g \, d\mu$. 
If $\theta_x \mu = \mu$ for any $x$ then $\mu$ is said to be translation invariant.

If $\Lambda $ is a finite subset of $\ZZ$ the marginal of $\mu \in {\mc P} (\Omega)$ 
on $\RR^{\Lambda}$ is denoted by $\mu |_{\Lambda}$. 
The relative entropy of $\nu \in {\mc P} (\Omega)$ with respect to $\mu \in {\mc P} (\Omega)$ 
in the box $\Lambda$ is defined by $H(\nu |_{\Lambda} \, | \, \mu |_{\Lambda} )$ 
and is denoted by $H_{\Lambda} (\nu| \mu)$. We say that a translation invariant probability 
measure $\nu \in {\mc P} (\Omega)$ has finite entropy density (with respect to $\mu$) 
if there exists a finite positive constant $C$ such that for any finite $\Lambda \subset \ZZ$, 
$H_{\Lambda} (\nu | \mu) \le C | \Lambda|$. In fact, if this condition is satisfied, then the limit 
\[
\overline{H} (\nu|\mu)=\lim_{|\Lambda| \to \infty} \frac{H_{\Lambda} (\nu | \mu) }{|\Lambda|}
\]
exists and is finite (see~\cite{FFL}). 
It is called the entropy density of $\nu$ with respect to $\mu$.  

We are now in position to define ergodicity.

\begin{definition}
\label{def:ergo}
We say that the infinite volume dynamics with infinitesimal generator ${\mc A}_{U,V}$ 
is \emph{ergodic} if the following claim is true:
If $\nu \in {\mc P} (\Omega)$ is a probability measure invariant by translation, 
invariant by the dynamics generated by ${\mc A}_{U,V}$ and with finite entropy density with respect to $\mu_{1,0,0}$, 
then $\nu$ is a mixture of the $\mu_{\beta,\lambda,\lambda'}, \beta>0, \lambda, \lambda' \in \RR$. 
\end{definition}

In the harmonic case ($U(z)=V(z)=z^2/2$) and for the Toda lattice ($U(z)=z^2/2$, $V(z) =e^{-z} +z -1$), the infinite system is completely integrable and an infinite number of conserved quantities can be explicitly written. It follows that they are not ergodic in the sense above. Nevertheless we expect that for a very large class of potentials, the Hamiltonian dynamics are ergodic and that these two cases are exceptional. 

In order that the infinite dynamics enjoy good ergodic properties, we superpose to the deterministic evolution a stochastic noise. 

Given a sequence $u= (u_y)_{y \in \ZZ} \in \RR^{\ZZ}$ and a site $x\in \ZZ$, we denote by $u^x$ (resp. $u^{x,x+1}$) the sequence defined by $(u^x)_y =u_y$ if $y \ne x$ and $(u^x)_x = -u_x$ (resp. $(u^{x,x+1})_y = u_y$ if $y \ne x,x+1$, $(u^{x,x+1})_{x}= u_{x+1}$ and $(u^{x,x+1})_{x+1} = u_x$). We consider the following noises (jump processes) whose generators are defined by their actions on functions $f: \Omega \to \RR$ according to:
\begin{enumerate}
\item $({\mc S}_{flip}^p f)(r,p)= \sum_{x} \left[ f(r,p^x) -f(r,p)\right]$.
\item  $({\mc S}_{flip}^r f)(r,p)= \sum_{x} \left[ f(r^x,p) -f(r,p)\right]$.
\item $({\mc S}_{ex}^p f)(r,p)= \sum_{x} \left[ f(r,p^{x,x+1}) -f(r,p)\right]$.
\item $({\mc S}_{ex}^r f)(r,p)= \sum_{x} \left[ f(r^{x,x+1},p) -f(r,p)\right]$.
\end{enumerate}

If $U$ is even then the noise ${\mc S}_{flip}^p$ conserves the energy, the deformation but not the momentum; if $U$ is odd the noise has little interest for us since the energy conservation is destroyed. Similarly, if $V$ is even the the noise ${\mc S}_{flip}^r$ conserves the energy and the momentum but not the deformation. The noises ${\mc S}_{ex}^p$ and ${\mc S}_{ex}^r$ conserve the energy, the deformation and the momentum. 

Let now $\gamma>0$ and denote by ${\mc L}$ the generator of the infinite Hamiltonian  dynamics generated by ${\mc A}_{U,V}$ perturbed by one of the previous noise ${\mc S}$ with intensity $\gamma$, i.e. ${\mc L}={\mc A}_{U,V} + \gamma {\mc S}$. 

\begin{theorem}[\cite{FFL}, \cite{BObook}, \cite{BS}]
The dynamics generated by ${\mc L}$ is ergodic in the sense that if $\nu \in {\mc P} (\Omega)$ is a probability measure invariant by translation, invariant by the dynamics generated by ${\mc L}$ and with finite entropy density with respect to $\mu_{1,0,0}$, then it holds:
\begin{enumerate}
\item If $U$ even and ${\mc S}={\mc S}_{flip}^{p}$ then $\nu$ is a mixture of the $\mu_{\beta,\lambda,0}$; 
\item If $V$ is even and ${\mc S}={\mc S}_{flip}^{r}$ then $\nu$ is a mixture of the $\mu_{\beta,0,\lambda'}$.
\item If ${\mc S}={\mc S}_{ex}^{p}$ or ${\mc S}={\mc S}_{ex}^{r}$ then $\nu$ is a mixture of the $\mu_{\beta,\lambda,\lambda'}$.  
\end{enumerate}
\end{theorem}

The main motivation to establish such a theorem is that by using Yau's relative entropy method (\cite{Y}) in the spirit of Olla-Varadhan-Yau (\cite{OVY}), it is possible to show that if the infinite volume dynamics is ergodic then the propagation of local equilibrium holds in the hyperbolic time scale, before the appearance of the shocks. As a consequence, the dynamics has a set of compressible Euler equations as hydrodynamic limits (\cite{BObook}, \cite{BS}). Observe that this is true also for the deterministic dynamics so that the rigorous derivation of the Euler equations from the first principles of the mechanics in the smooth regime is ``reduced'' to prove that the dynamics generated by ${\mc A}_{U,V}$ is ergodic.

\section{Harmonic chains }
\label{sec:harmonic}

\subsection{Role of the conservation of momentum and deformation}

We consider here the specific (harmonic) case $V(z)=U(z)=z^2/2$. The dynamics is then linear and can be solved analytically using Fourier transform. Let us introduce a new macroscopic variable $\eta \in \RR^{\ZZ}$ defined from $(p,r) \in \Omega$ by setting 
\begin{equation}
\label{eq:eta}
\eta_{2x}=r_x, \quad \eta_{2x+1}=p_{x+1}, \quad x\in \ZZ.
\end{equation}
Then, the Hamiltonian dynamics can be rewritten in the form
\begin{equation}
\label{eq:etadyn}
\frac{d\eta_x}{dt} = V'(\eta_{x+1}) -V' (\eta_{x-1}), \quad x \in \ZZ.
\end{equation}

We introduce the $k$th mode ${\widehat \eta} (k, \cdot)$ for $k \in {\bb T} = 
\mathbb{R}/\ZZ$, the one-dimensional torus of length $1$: 
\begin{equation*}
{\widehat \eta} (t,k) =\sum_{x \in \ZZ} \eta_x (t) \, e^{2 \ri \pi k x}. 
\end{equation*}
Then, the equations of motion are equivalent in the sense of distributions to the 
following decoupled system of first order differential equations:
\begin{equation*}
\frac{d{\widehat \eta}}{dt} (t,k)  = \ri \omega (k) \, {\widehat \eta} (t,k), 
\end{equation*}
where the dispersion relation $\omega (k)$ reads
\begin{equation*}
\omega (k) =- 2 \sin (2\pi k),
\end{equation*} 
and the group velocity $v_{\rm g}$ is
\begin{equation*}
v_{\rm g} (k)  = \omega' (k) = -4\pi \cos (2\pi k).
\end{equation*}
By inverting the Fourier transform, the solution can be written as
\[
\eta_{x} (t) = \int_{\TT} {\widehat \eta} (t,k) \, e^{- 2 \ri \pi k x} \, dk.
\]
If the initial configuration $\eta (0)$ is in $\ell_2$ the well defined energy of the $k$th mode
\[
E_k (t)= \frac{1}{4\pi} |{\widehat \eta} (t,k)|^2 = E_k (0)
\]
is conserved by the time evolution, and the total energy current 
${\tilde J}^e= \sum_{x \in \ZZ} j^{e}_{x,x+1}$ takes the simple form
\begin{equation*}
{\tilde J}^e = \int_{\TT} v_{\rm g}(k) E_k \, dk.
\end{equation*}

We interpret the waves ${\widehat \eta} (k,t)$ as fictitious particles (phonons in solid state physics). In the absence of nonlinearities, they travel the chain without scattering. The diffusion of energy is then said to be ballistic. If the potential is non-quadratic, it may be expected that the nonlinearities produce a scattering responsible for the diffusion of the energy. Nevertheless, the conservation of the deformation and of the momentum implies that $\sum_x (r_x +p_x)$ is conserved
\begin{equation}
\label{eq:consj}
{\widehat \eta} (t, 0) = {\widehat \eta} (0,0).
\end{equation}
The identity (\ref{eq:consj}) is valid even if $U \ne V$ and $U,V$ are not quadratic. It means that the $0$th mode is not scattered at all and crosses the chain ballistically. In fact, the modes with small wave number $k$ do not experience a strong scattering and they therefore contribute to the observed anomalous diffusion of energy.

It is usually explained that momentum conservation plays a major role in the anomalous diffusion of energy but it is clear that the deformation conservation plays exactly the same role as momentum and that it is the conservation of their sum which is the real ingredient producing anomalous diffusion of energy (see Theorem \ref{th:harmflip} and Theorem \ref{th:anharmflip}). 

\subsection{Green-Kubo formula}

The signature of an anomalous diffusion of energy can be seen at the level of the Green-Kubo formula. When transport of energy is normal, meaning that the macroscopic equations such as the Fourier's law or heat equation hold, the transport coefficient appearing in these equations can be expressed by the famous Green-Kubo formula. In order to define the latter we need to introduce some notations. Since the discussion about the Green-Kubo formula is not restricted to the harmonic case we go back to a generic anharmonic model in the rest of the Subsection.

Recall that the probability measures $\mu_{\beta, \lambda, \lambda'}$ form a family 
of invariant probability measures for the infinite dynamics generated by ${\mc A}_{U,V}$. The following thermodynamic relations (which are valid 
since we assumed that the partition function $Z$ is well defined on $(0,+\infty) \times {\mathbb R} \times {\mathbb R}$) relate the chemical potentials $\beta, \lambda, \lambda'$ 
to the mean energy $e$, the mean deformation $u$, the mean momentum $\pi$ under $\mu_{\beta,\lambda,\lambda'}$:
\begin{eqnarray}
\label{eq:tr}
&&e(\beta,\lambda, \lambda')= \mu_{\beta,\lambda,\lambda'} (U(p_x) +V(r_x))= -\partial_{\beta}\Big(\log Z(\beta,\lambda, \lambda')\Big),\\
&&u(\beta,\lambda,\lambda') =\mu_{\beta,\lambda,\lambda'} (r_x)= -\partial_{\lambda} \Big(\log Z(\beta,\lambda, \lambda')\Big), \\
&&\pi(\beta,\lambda,\lambda') =\mu_{\beta,\lambda,\lambda'} (p_x)= -\partial_{\lambda'} \Big(\log Z(\beta,\lambda,\lambda')\Big).
\end{eqnarray} 
These relations can be inverted by a Legendre transform to express $\beta$, $\lambda$ and $\lambda'$ 
as a function of~$e$, $u$ and~$\pi$. 
Define the thermodynamic entropy $S:(0,+\infty) \times \RR \times \RR \to [0,+\infty]$ as
\begin{equation*}
S(e,u,\pi)= \inf_{ \lambda, \lambda' \in \RR^2, \beta >0} \Big\{ \beta e + \lambda u + \lambda' \pi 
+ \log Z (\beta,\lambda, \lambda') \Big\}.
\end{equation*}
Let ${\mc U}$ be the convex domain of $(0,+\infty) \times \RR\times \RR$ where $S(e,u,\pi) <+ \infty$ and $\mathring{\mc U}$ its interior. Then, for any $(e,u,\pi):=(e(\beta, \lambda,\lambda'),u(\beta,\lambda, \lambda'), \pi(\beta,\lambda,\lambda')) 
\in {\mathring{\mc U}}$, the parameters $\beta,\lambda,\lambda'$ can be obtained as 
\begin{equation}
\label{eq:14}
\beta = (\partial_e S ) (e,u,\pi), 
\qquad 
\lambda= (\partial_u S) (e,u,\pi),
\qquad
\lambda'= (\partial_\pi S)(e,u,\pi)
\end{equation} 

These thermodynamic relations allow us to parameterize the Gibbs states by the average values of the conserved quantities $(e,u,\pi)$ rather than by the chemical potentials $(\beta,\lambda,\lambda')$. Thus, we denote by $\nu_{e,u,\pi}$ the Gibbs measure $\mu_{\beta,\lambda,\lambda'}$ where $(e,u,\pi)$ are related to $(\beta,\lambda,\lambda')$ by (\ref{eq:14}). Let $J^e:=J^e (e,u,\pi)= \nu_{e,u,\pi} (j^e_{x,x+1})$ be the average of the energy current $j^{e}_{x,x+1}= -U' (p_x) V'(r_x)$ and define the normalized energy current ${\hat j}^e_{x,x+1}$ by 
\begin{equation*}
{\hat j}^e_{x,x+1} = j^e_{x,x+1} - J^e  - (\partial_e J^e) (e_x -e) -(\partial_u J^e) (r_x -u) - (\partial_\pi J^e) (p_x -\pi).
\end{equation*}
The normalized energy current  is the part of the centered energy current which is orthogonal in ${\LL}^2 (\nu_{e,u,\pi})$ to the space spanned by  the conserved quantities.

Up to multiplicative thermodynamic parameters (see \cite{Sp} for details) that we neglect to simplify the notations, the Green-Kubo formula {\footnote{The transport coefficient is in fact a matrix whose size is the number of conserved quantities. Since we are interested in the energy diffusion, we only consider the entry corresponding to the energy-energy flux.}} is nothing but
\begin{equation*}
\kappa (e,u,\pi) := \int_0^{\infty} \; \sum_{x \in \ZZ} {\mathbb E}_{\nu_{e,u,\pi}} \left[ {\hat j}^e_{x,x+1} (\omega(t)) \; {\hat j}_{0,1}^e (\omega(0)) \right] \; dt
\end{equation*}
where ${\mathbb E}_{\nu_{e,u,\pi}}$ denotes the expectation corresponding to the law of the infinite volume dynamics $(\omega(t))_{t\ge 0}$ generated by ${\mc A}_{U,V}$ with initial condition $\omega (0)$ distributed according to the equilibrium Gibbs measure $\nu_{e,u,\pi}$.  The definition of $\kappa (e,u,\pi)$ is formal but the way we adopt to give it a mathematically well posed definition is to introduce a small parameter $z>0$ and define $\kappa (e,u,\pi)$ as
\begin{equation}
\label{eq:GK}
\kappa (e,u,\pi) = \limsup_{z \to 0} \ll {\hat j}^e_{0,1}\, , \, (z-{\mc A}_{U,V})^{-1} {\hat j}^e_{0,1} \gg_{e,u,\pi} 
\end{equation}
where the inner-product $\ll \cdot, \cdot \gg_{e,u,\pi}$ is defined for local square integrable functions $f,g : \Omega \to \RR$ by 
\begin{equation*}
\ll f \, , \, g \gg_{e,u,\pi} \; = \; \sum_{x \in \ZZ} \left[ \left( \int f \theta_x g \, d\nu_{e,u, \pi} \right) - \left(\int f d\nu_{e,u,\pi} \right) \left(\int g d\nu_{e,u,\pi} \right) \right].  
\end{equation*} 
Since $(z-{\mc A}_{U,V})^{-1} {\hat j}^e_{0,1}$ is not a local function, the term on the RHS of (\ref{eq:GK}) has to be interpreted in the Hilbert space obtained by the completion of the space of local bounded functions with respect to the inner product $\ll \cdot, \cdot \gg_{e,u,\pi}$. 

The superdiffusion (resp. normal diffusion) of energy corresponds to an infinite (resp. finite) value for $\kappa (e,u,\pi)$. In order to study the superdiffusion, it is of interest to estimate the time decay of the autocorrelation of the normalized current 
$$C (t):= C_{e,u,\pi} (t) = \sum_{x \in \ZZ}  {\mathbb E}_{\nu_{e,u,\pi}} \left[ {\hat j}^e_{x,x+1} (\omega(t)) \; {\hat j}_{0,1}^e (\omega(0)) \right].$$
It is in general easier to estimate the behavior of the Laplace transform $L(z)=\int_{0}^{\infty} e^{-z t} C(t) dt$ as $z\to 0$. Roughly, if  $L(z) \sim z^{- \delta}$ for some $\delta \ge 0$ then $C(t) \sim t^{\delta -1}$ as $t \to + \infty$. Observe also that  
\begin{equation*}
L(z)= \ll {\hat j}^e_{0,1}\, , \, (z-{\mc A}_{U,V})^{-1} {\hat j}_{0,1} \gg_{e,u,\pi}.
\end{equation*}

\subsection{Harmonic chain perturbed by a conservative stochastic noise}

We consider now the particular case $U(z)=V(z)=z^2 /2$ and study the Green-Kubo formula for the perturbed dynamics generated by ${\mc L} ={\mc A}_{U,V} +\gamma {\mc S}$ where ${\mc S}$ is one of the noises introduced in Section \ref{sec:ergodicity}. Since, depending of the form of the noise, the momentum conservation law (resp. deformation conservation law) can be suppressed, the corresponding Green-Kubo formula shall be modified by setting $\pi=0$ and $\partial_\pi J^e=0$ (resp. $u=0$ and $\partial_u J^e =0$).

We have the following theorem which shows that if momentum conservation law or deformation conservation law is destroyed by the noise then a normal behavior occurs.

\begin{theorem}[\cite{BO}]
\label{th:harmflip}
Let $U$ and $V$ be quadratic potentials. 
\begin{enumerate}
\item Consider the system generated by ${\mc L}={\mc A}_{U,V} + \gamma {\mc S}_{flip}^p$, $\gamma>0$. Then the following limit
\begin{equation*}
\lim_{z \to 0} \ll {\hat j}^e_{0,1}\, , \, (z-{\mc L})^{-1} {\hat j}^e_{0,1} \gg_{e,u, 0} 
\end{equation*}
exists, is finite and strictly positive and can be explicitly computed. 
 \item Consider the system generated by ${\mc L}={\mc A}_{U,V} + \gamma {\mc S}_{flip}^r$, $\gamma>0$. Then the following limit
\begin{equation*}
\lim_{z \to 0} \ll {\hat j}^e_{0,1}\, , \, (z-{\mc L})^{-1} {\hat j}^e_{0,1} \gg_{e,0,\pi} 
\end{equation*}
exists, is finite and strictly positive and can be explicitly computed. 
\end{enumerate}
\end{theorem}  

It shall be noticed that the second statement is a direct consequence of the first one since the process of the second item is equal to the first one by the transformation
\begin{equation*}
r_x \to p_x,\quad p_{x} \to r_{x-1}.
\end{equation*}
However, the interest of the second statement is to show that {\textit{even if momentum is conserved}}, a normal diffusion of energy occurs. This is because the deformation is no longer conserved. 

The following theorem shows that if the noise added conserves momentum {\textit{and}} deformation then the situation is very different since an anomalous diffusion of energy is observed. 

\begin{theorem}[\cite{BBO1},\cite{BBO2}]
Let $U$ and $V$ be quadratic potentials. 
\begin{enumerate}
\item Consider the system generated by ${\mc L}={\mc A}_{U,V} + \gamma {\mc S}_{ex}^p$, $\gamma>0$. Then the following limit
\begin{equation*}
\lim_{z \to 0} z^{1/2} \; \ll {\hat j}^e_{0,1}\, , \, (z-{\mc L})^{-1} {\hat j}^e_{0,1} \gg_{e,u, \pi} 
\end{equation*}
exists, is finite and strictly positive and can be explicitly computed. 
 \item Consider the system generated by ${\mc L}={\mc A}_{U,V} + \gamma {\mc S}_{ex}^r$, $\gamma>0$. Then the following limit
\begin{equation*}
\lim_{z \to 0} z^{1/2} \ll {\hat j}^e_{0,1}\, , \, (z-{\mc L})^{-1} {\hat j}^e_{0,1} \gg_{e,u,\pi} 
\end{equation*}
exists, is finite and strictly positive and can be explicitly computed. 
\end{enumerate}
In particular, in each of the previous case the Green-Kubo formula yields an infinite conductivity. 
\end{theorem}

\section{ Anharmonic chains}

We consider now the anharmonic case. For deterministic chains generated by ${\mc A}_{U,V}$ we expect usually a superdiffusive behavior of the energy. If a noise ${\mc S}$ is superposed to the dynamics, we expect that transport is normal for ${\mc S}={\mc S}_{flip}^p$ and ${\mc S}={\mc S}_{flip}^r$ and superdiffusive if ${\mc S}={\mc S}_{ex}^p$ or ${\mc S}={\mc S}_{ex}^r$.

The following theorem generalizes Theorem \ref{th:harmflip} to the anharmonic case showing that a  noise destroying momentum conservation law or deformation conservation law produces normal transport. This shows that, also in the anharmonic case, momentum conservation alone is not responsible of anomalous diffusion of energy but that deformation conservation law plays a similar role.

\begin{theorem}[\cite{BO}]
\label{th:anharmflip}
Let $U$ and $V$ be smooth potentials such that there exists a constant $c>0$ such that
\begin{equation*}
c \le U '' \le c^{-1}, \quad c \le V'' \le c^{-1}.
\end{equation*} 

\begin{enumerate}
\item Assume $U$ even and consider the system generated by ${\mc L}={\mc A}_{U,V} + \gamma {\mc S}_{flip}^p$, $\gamma>0$. Then the following limit
\begin{equation*}
\lim_{z \to 0} \ll {\hat j}^e_{0,1}\, , \, (z-{\mc L})^{-1} {\hat j}^e_{0,1} \gg_{e,u,0} 
\end{equation*}
exists and is finite. 
 \item Assume $V$ even and consider the system generated by ${\mc L}={\mc A}_{U,V} + \gamma {\mc S}_{flip}^r$, $\gamma>0$. Then the following limit
\begin{equation*}
\lim_{z \to 0} \ll {\hat j}^e_{0,1}\, , \, (z-{\mc L})^{-1} {\hat j}^e_{0,1} \gg_{e,0,\pi} 
\end{equation*}
exists and is finite.
\end{enumerate}
\end{theorem}

\begin{proof}
The second statement is a direct consequence of the first one by the symmetry argument evoked for Theorem \ref{th:harmflip}. The upper bounds on $U''$ and $V''$ are here to assure the existence of the infinite volume dynamics. 

For simplicity assume that $u=0$ and $\beta:=\beta(e,u,0)=1$. The first statement has been proved in \cite{BO} in the particular case $U(z)=z^2/2$. The generalization to a non quadratic smooth even potential $U$ is straightforward. In \cite{BO}, since $U(z)=z^2/2$, we used Hermite polynomials which are orthogonal w.r.t. the Gaussian measure $d\mu (z) = (2\pi)^{-1/2} \exp\{ -z^2 /2\} dz$. In the present case, the only difference is that we have to replace the Hermite basis by any orthogonal polynomial basis $\{P_n \}_{n \ge 0}$ with respect to the probability measure ${\mc N}^{-1} \exp (- U(z) ) dz$ (with ${\mc N}$ a normalization constant) which satisfies $P_n$ odd if $n$ odd and even otherwise. Then the proof is exactly the same.
\end{proof}

It would be now of interest to show that if we perturb the dynamics generated by ${\mc A}_{U,V}$ by ${\mc S}_{ex}^p$ or by ${\mc S}_{ex}^r$ then anomalous diffusion of energy occurs {\footnote{However, if $U$ or $V$ is bounded, like for the rotors model, we expect that diffusion is normal. }}. This is an open question and as far as we know the only result going in this direction has been obtained in \cite{BG}.   

The model considered in \cite{BG} is the dynamics generated by ${\mc A}_{U,V}$ with $U=V$ taking the particular form $V(z)=e^{-z} +z -1$, perturbed by a noise ${\mc S}$ which conserves energy and $\sum_{x \in \ZZ} (r_x + p_x)$. More exactly, let us rewrite the Hamiltonian dynamics (\ref{eq:Ham-gen}) by using the variable $\eta:=(\eta_x)_{x \in \ZZ} \in \RR^{\ZZ}$ defined by (\ref{eq:eta}). Then we get the equations of motion given by (\ref{eq:etadyn}). With these new variables, the total energy is $2 \sum_{x} V(\eta_x)$, the total deformation is $\sum_{x} \eta_{2x}$ and the total momentum is $\sum_{x} \eta_{2x+1}$. The noise ${\mc S}$ superposed to the dynamics acts on local functions $f: \RR^{\ZZ} \to \RR$ according to 
\begin{equation*}
({\mc S} f) (\eta) = \sum_{x \in \ZZ} \left[ f(\eta^{x,x+1}) -f(\eta) \right].
\end{equation*}  

Observe that the noise conserves the energy, destroys the momentum and the deformation conservation laws but conserves $\sum_{x} \eta_x = \sum_{x} (p_x +r_x)$, which as explained above is the quantity (that we call the ``volume'' to follow the terminology used in \cite{BG}) responsible of the anomalous diffusion of energy. Since we have now only two conserved quantities (the energy and the volume), the Gibbs states of the perturbed dynamics are given by $\{ \mu_{\beta, \lambda, \lambda} \}_{\beta>0, \lambda}$ or equivalently by $\{ \nu_{e, \pi, \pi} \, ; \, e>0, \pi \}$. The normalized energy current is given by
\begin{equation*}
{\hat j}^e_{x,x+1} (\eta)= -2 V' (\eta_x) V' (\eta_{x+1}) +2 \tau^2 +2 \partial_e (\tau^2) \, (2V(\eta_x) -e) + 2 \partial_\pi (\tau^2) \, (\eta_x -\pi)
\end{equation*}
with $\tau:=\tau(e,\pi) = \int V'(\eta_x) d\nu_{e,\pi,\pi}$. 

\begin{theorem}[\cite{BG}]
Let $(e,\pi) \in (0,+\infty) \times \RR$ such that $\nu_{e,\pi,\pi}$ is well defined. Consider the dynamics with generator ${\mc L}={\mc A}_{exp} +\gamma {\mc S}$, $\gamma>0$, where 
\begin{equation}
{\mc A}_{exp} = \sum_x (V'(\eta_{x+1}) -V' (\eta_{x-1})) \partial_{\eta_x},
\end{equation}
and $V(z)=e^{-z}+z-1$. Then there exists a constant $c>0$ such that for any $z>0$ 
\begin{equation*}
c z^{-1/4} \; \le \; \ll {\hat j}^e_{0,1} , (z-{\mc L})^{-1} {\hat j}^e_{0,1} \gg_{e,\pi,\pi} \le c^{-1} z^{-1/2}.
\end{equation*} 
It follows that the Green-Kubo formula of the energy transport coefficient yields an infinite value.
\end{theorem}

We expect that the system above belongs to the KPZ class so that $ \ll {\hat j}^e_{0,1} , (z-{\mc L})^{-1} {\hat j}^e_{0,1} \gg_{e,\pi,\pi}$ should diverge like $z^{-1/3}$. In the present state of the art no robust technique is available to show such result apart from the non-rigorous (but powerful) mode-coupling theory (\cite{MS}, \cite{Sp2}, \cite{VB}). A second open problem is to generalize the previous theorem to other interaction potentials $V$. Numerical simulations have been reported in \cite{BS}.

%
%
\input{Bernardinref}

\end{document}

%% file: Bernardinref.tex
%
%
%


%% file: Bernardin.bbl
\begin{thebibliography}{99.}



\bibitem{BBO1} G. Basile, C. Bernardin, S. Olla, Momentum conserving model with anomalous thermal conductivity in low dimensional systems, Phys. Rev. Lett. {\bf 96} (2006), 204303.

\bibitem{BBO2} G. Basile, C. Bernardin, S. Olla, Thermal conductivity for a momentum conserving model, Comm. Math. Phys. {\bf 287} (2009), no. 1, 67--98.

%
%
%
%
\bibitem{BG} C. Bernardin and P. Gon\c alves, Anomalous fluctuations for a perturbed Hamiltonian system with exponential interactions, to appear in Commun. Math. Phys. (2013).

\bibitem{BO} C. Bernardin and S. Olla, Transport Properties of a Chain of Anharmonic Oscillators with
  random flip of velocities, J. Stat. Phys. {\bf 145}, (2011).

\bibitem{BObook} C. Bernardin and S. Olla, Non-equilibrium macroscopic dynamics of chains of anharmonic oscillators, in preparation, available at {\tt http://perso.ens-lyon.fr/cedric.bernardin/Files/springs.pdf}
%
\bibitem{BS} C. Bernardin and G. Stoltz, Anomalous diffusion for a class of systems with two conserved quantities, Nonlinearity {\bf 25}, Num. 4, 1099–1133.



\bibitem{BLR} F. Bonetto, J.L. Lebowitz, Rey-Bellet, Fourier's law: A challenge to theorists, in Mathematical Physics 2000, A. Fokas et al. (eds.), Imperial College Press, London, (2000), 128--150.



\bibitem{D} A. Dhar, Heat Transport in low-dimensional systems, Adv. Phys.,  {\bf 57} (2008), 457.


%
%

%
%
%
%
%
\bibitem{FFL} J. Fritz, T. Funaki and J.L. Lebowitz, Stationary states of random Hamiltonian systems, Probab. Theory Related Fields, {\bf 99} (1994), 211--236.



\bibitem{GDL} A. Gerschenfeld, B. Derrida and J.L. Lebowitz, Anomalous FourierÕs Law and Long Range Correlations in a 1D Non-momentum Conserving Mechanical Model, J. Stat. Phys. (2010) {\bf 141}, 757--766.


%
%
%
%
%



%
\bibitem{LP} J.L. Lebowitz, O. Penrose, Modern Ergodic Theory, Physics Today {\bf 26}, (1973), 155--175.  

\bibitem{LLP} S. Lepri, R. Livi, A. Politi,  Thermal Conduction in classical low-dimensional lattices, Phys. Rep. \textbf{377} (2003), 1--80.

%

\bibitem{MS} C. B. Mendl, H. Spohn, Dynamic Correlators of FPU Chains and Nonlinear Fluctuating Hydrodynamics, (2013), arXiv:1305.1209.


%
\bibitem{OVY} {{S. Olla, S.R.S. Varadhan and H.T. Yau}}, {Hydrodynamic Limit for a Hamiltonian System with Weak Noise}, Commun. Math. Phys. {\bf 155} (1993), 523--560.
%
%
%
%
%

%
%
%
%
\bibitem{Sp}
H. Spohn, Large Scale Dynamics of Interacting Particles, Springer (1991).

\bibitem{Sp2}
H. Spohn, Nonlinear fluctuating hydrodynamics for anharmonic chains, \url{arXiv:1305.6412} (2013).

\bibitem{Sz} D. Sz\'asz, Boltzmann's ergodic hypothesis, a conjecture for centuries? Hard ball systems and the Lorentz gas,, Encyclopaedia Math. Sci., {\bf 101}, Springer (2000), 421--448. 


%
%

%
%
%
\bibitem{VB} H. van Beijeren, Exact results for anomalous transport in one dimensional Hamiltonian systems, Phy. Rev. Let.,  \textbf{28}   (2012).

\bibitem{Y} H.T. Yau, Relative entropy and hydrodynamics of Ginzburg-Landau models. Lett. Math. Phys. {\bf 22} (1991), no. 1, 63--80

%
%



\end{thebibliography}
